\documentclass[10pt,conference]{IEEEtran}
\IEEEoverridecommandlockouts
\usepackage{authblk}
\usepackage{cite}
\usepackage{amsmath,amssymb,amsfonts,dsfont}
\usepackage{algorithm}
\usepackage[noend]{algorithmic}
\usepackage{graphicx}
\usepackage{textcomp}
\usepackage{xcolor}
\usepackage{bm}
\usepackage{booktabs}
\usepackage{multirow}
\usepackage{lipsum}
\usepackage{url}
\usepackage{amsthm}
\usepackage{geometry}
\ifCLASSOPTIONcompsoc
\usepackage[caption=false, font=normalsize, labelfont=sf, textfont=sf]{subfig}
\else
\usepackage[caption=false, font=footnotesize]{subfig}

\newtheorem{theorem}{Theorem}

\newtheorem{remark}{Remark}

\begin{document}

\title
{
Neural Estimation of the Information Bottleneck \\ Based on a Mapping Approach
\thanks{The first three authors contributed equally to this work and $\dag$ marked the corresponding author. The authors would like to thank Prof. Hao Wu for his valuable discussions and helpful advice. This work was partially supported by National Key Research and Development Program of China (2018YFA0701603) and National Natural Science Foundation of China (12271289 and 62231022).}
}


\author[1]{Lingyi Chen}
\author[1]{Shitong Wu}
\author[3]{Sicheng Xu}
\author[4$\dag$]{Wenyi Zhang}
\author[2]{Huihui Wu}

\affil[1]{Department of Mathematical Sciences, Tsinghua University, Beijing 100084, China}
\affil[2]{Yangtze Delta Region Institute (Huzhou), 
\authorcr University of Electronic Science and Technology of China, 
Huzhou, Zhejiang, 313000, China. 
} 
\affil[3]{Department of Mathematical Sciences,
\authorcr University of Science and Technology of China, Hefei, Anhui 230027, China 
}
\affil[4]{Department of Electronic Engineering and Information Science, 
\authorcr University of Science and Technology of China, Hefei, Anhui 230027, China 
\authorcr Email: wenyizha@ustc.edu.cn 
}

\maketitle

\begin{abstract}
The information bottleneck (IB) method is a technique designed to extract meaningful information related to one random variable from another random variable, and has found extensive applications in machine learning problems. 
In this paper, neural network based estimation of the IB problem solution is studied, through the lens of a novel formulation of the IB problem. 
Via exploiting the inherent structure of the IB functional and leveraging the mapping approach, the proposed formulation of the IB problem involves only a single variable to be optimized, and subsequently is readily amenable to data-driven estimators based on neural networks. 
A theoretical analysis is conducted to guarantee that the neural estimator asymptotically solves the IB problem, and the numerical experiments on both synthetic and MNIST datasets demonstrate the effectiveness of the neural estimator. 
\end{abstract}
\begin{IEEEkeywords}
Information bottleneck, mapping approach, neural network estimation. 
\end{IEEEkeywords}

\section{Introduction}

The information bottleneck (IB) method, first proposed by Tishby et al. in 1999 \cite{tishby2000information}, provides an information theoretic framework on characterizing the best trade-off between information compression and extraction.
The main goal of IB is to extract information of a target variable $Y$ from another observable variable $X$ given the joint distribution.
Due to its capability to preserve meaningful information while   compressing representations, it has been applied to various fields, such as clustering \cite{slonim2000document},  source coding \cite{courtade2011multiterminal} and supervised learning \cite{alemi2016deep,tishby2015deep}, and so on. 
In recent years, the IB method has gained prominence in neural network research.
For example, work  \cite{alemi2016deep} and \cite{achille2018information} apply IB to an encoder-decoder structure for image classification. 
Furthermore, other studies (e.g., \cite{tishby2015deep,saxe2019information}) have used the IB principle to theoretically analyze deep neural networks.
Applying IB to neural networks offers benefits, especially for high-dimensional datasets, as it enhances data representation extraction.
However, existing neural network methods to solve the IB problems, represented by the variational information bottleneck (VIB) \cite{alemi2016deep} etc., have certain limitations, due to the fact that they adopt the classical problem formulation which involves subtracting two mutual information terms. 
To avoid the difficulty of directly computing mutual information\cite{belghazi2018mutual}, VIB relaxed the original IB formulation and turned to optimize a variational bound of the IB problem. 
As a cost, it fails to give a tight enough estimation of the IB problem \cite{kolchinsky2018caveats}. 
Therefore, it is demanding to develop an equivalent and neural network-friendly formulation for the IB problem. 
In this paper, we propose a novel formulation for the IB problem that can be parameterized by neural networks without any relaxations. 
We approach the problem from a mapping perspective.
The mapping approach, as suggested by Rose \cite{rose1994mapping}, is a powerful technique to transform the problem of optimizing a probability measure into the problem of finding an optimal mapping, and it has been successfully applied to the computation and analysis of rate-distortion functions. 
Specifically, we first obtain the dual form of the IB problem, and then transform the optimization problem of one probability measure into an optimal mapping search problem \cite{rose1994mapping}. 
Since the dual form possesses two variable probability measures, the application of mapping approach into IB problem is more involved than the rate-distortion case. 
Remarkably, we discover that when extending the mapping approach to the IB problem, the two variables of the dual form can be combined into one variable. 
As a result, the final univariate formulation is obtained. 
Based on the proposed formulation, we designed a neural estimation algorithm. 
Specifically, the variable is represented by neural networks, and the optimization objective is computed as the loss function for gradient descent. 
Moreover, an analysis of the neural estimation algorithm is conducted, and it reveals that the solution obtained by the neural estimation converges to the theoretical solution of the original IB problem, when the sample size is sufficiently large. 
Simulation results demonstrate that in the cases of the low-dimensional sources, the proposed method is effective and accurate compared with theoretical results. 
As for the cases of high-dimensional dataset, such as the MNIST dataset \cite{lei2023neural}, the proposed method provides a tighter estimation than the VIB approach. 
The remaining of the article is organized as follows. 
In Section \uppercase\expandafter{\romannumeral2}, a brief review of the background knowledge related to the IB problem is provided.
Next, we derive the equivalent formulation of the IB problem using mapping approach in Section \uppercase\expandafter{\romannumeral3}.
Section \uppercase\expandafter{\romannumeral4} presents a neural estimation algorithm based on the proposed formulation, along with an analysis of its asymptotic properties.
In Section \uppercase\expandafter{\romannumeral5}, experiments are conducted to evaluate the effectiveness of the algorithm.
Finally, the paper concludes in Section \uppercase\expandafter{\romannumeral6}.

\section{Preliminaries}
The information bottleneck (IB) is a method to extract information of a target variable $Y$ from another observable variable $X$.
The extracted information is represented by a bottleneck variable $T$, which forms a Markov chain $Y \!\leftrightarrow\! X \!\leftrightarrow\! T$.
Specifically, given the joint distribution $P_{XY}$, the IB problem is defined by
\begin{equation}
    \begin{aligned}
        R(I)&=\min_{P_{T|X}} I(X;T)\\
        & s.t. \quad I(Y;T)\geq I,
    \end{aligned}
    \label{IB}
\end{equation}
where the mutual information $I(X; T)$ reflects the effectiveness of compressing $X$, and the mutual information $I(Y; T)$ signifies how well $T$ captures the essence of $Y$. 
Consider that the variables $X,T,Y$ take values from alphabets $\mathcal{X},\mathcal{T},\mathcal{Y}$ with probability densities $p(x),p(t),p(y)$ respectively, and the conditional probability densities $p(t|x),p(y|x)$ etc. are defined similarly. 
The expression of $I(Y;T)$ is 
$$
I(Y;T)\!=\!\!\int_{\mathcal{X}}\!dx\,p(x)\int_\mathcal{T}\! dt\,p(t| x)\int_\mathcal{Y} \!dy\,p(y| x)\log\frac{p(y| t)}{p(y)},
$$
where the variable $p(y|t)$ is defined by the variable $p(t|x)$, i.e., $p(y| t)=\frac{\int_{\mathcal{X}}dx\,p(x)p(y| x)p(t| x)}{\int_{\mathcal{X}}dx\,p(x)p(t| x)}$. 
It is obvious that by substituting $p(y|t)$ into the expression of $I(Y;T)$, the IB problem \eqref{IB} will involve  a complex combination of integral terms, making the computation difficult \cite{belghazi2018mutual}. 
%

%
To address this problem, work \cite{tishby2000information} relaxed  $p(t),p(y|t)$ to independent densities $q(t)$ and $r(y|t)$ respectively, and proposed the following triple minimization formulation of IB: 
\begin{equation}
    \begin{aligned}
        \min_{p(t| x)}&\min_{q(t)}\min_{r(y| t)}\mathcal{L}(p(t| x),q(t),r(y| t))= \\
        \int_\mathcal{X} & dx\,p(x)\int_\mathcal{T} dt\,p(t| x)\log \left[{p(t| x)}/{q(t)}\right] \\
        -\beta\!\!\int_\mathcal{X}& dx\,p(x)\int_\mathcal{T} dt\,p(t| x)\int_\mathcal{Y} dy\,p(y| x)\log r(y| t). 
    \end{aligned}
    \label{tri}
\end{equation}
This formulation is equivalent to the IB problem \eqref{IB} but without the  distribution constraints. 
The multiplier $\beta$ corresponds to the slope of the IB curve \cite{zaidi2020ib_view} and BA algorithm of the IB problem is derived based on this formulation \cite{tishby2000information}. 
For the IB problem, the existing methods that adopt neural networks are usually designed based on the above formulation \cite{zhai2022adversarial,alemi2016deep,kolchinsky2019nonlinear_ib}. 
%
%
However, the usage of this formulation \eqref{tri} may lead to some unnecessary relaxations of the IB problem when applying neural networks, since it involves subtraction of two mutual information terms that is difficult to be directly represented by neural networks without variational relaxations. 
Hence, these methods usually fail to provide accurate enough solutions to the IB problems.  
To conquer this difficulty, this paper reformulates the IB problem \eqref{IB} from the perspective of mapping, as described in the next section. 
%

\section{A Mapping Approach to the Information Bottleneck Model}

This section derives a novel formulation equivalent to the original IB problem, by reformulating the triple minimization problem \eqref{tri} via the mapping approach. 
The advantage of the proposed formulation is that it only processes one variable without complicated operations, which is suitable for neural representation. 
%
Since mapping approach (MA) is the main technique, we name the novel formulation MA-IB model. 
The definition of the MA-IB model is presented in Theorem 1. 

%
\begin{theorem}[MA-IB Model]
Suppose $Z\in\mathcal{Z}$ is a given random variable with probability measure $P_Z$, and $r(y|z)$ represents the conditional probability density from $Z$ to $Y$. 
Then the following formulation \eqref{vmaib} shares the same optimal values with the triple minimization problem \eqref{tri}: 
\begin{equation}
  \begin{aligned}
        &\min_{r(y| z)}G(r(y| z))=\\
        -&\!\!\int_\mathcal{X} \!\!dx\, p(x)\!\log\left(\!\int_\mathcal{Z} \!\!dz\,p(z)\,e^{\beta \int_\mathcal{Y} dy\,p(y| x)\log r(y| z)}\!\right).
    \end{aligned}
\label{vmaib}
\end{equation}
\label{thm1}
\end{theorem}

\begin{proof}
According to \cite{tishby2000information}, the following optimal condition can be obtained through variational calculations, 
\begin{equation} 
    p(t| x)=\frac{q(t)e^{\beta\!\int_\mathcal{Y}\!dy\,p(y| x)\log\! r(y| t)}}{A(x)},
    \label{ptx}
\end{equation}
where the constant $A(x)=\int_\mathcal{T} dt\,q(t)e^{\beta\int_\mathcal{Y}dy\,p(y| x)\log r(y| t)}$ is the normalization factor \cite{tishby2000information} to guarantee $\int_\mathcal{T}dt\,p(t| x)=1$.
By substituting \eqref{ptx} into the formulation \eqref{tri}, one obtains another equivalent formulation of the IB problem:
\begin{equation}
    \begin{aligned}
        &\min_{q(t)}\!\min_{r(y| t)} F(q(t),r(y| t))=\\
        -&\int_\mathcal{X} \!\!dx\, p(x)\log\left(\int_\mathcal{T} \!\!dt\,q(t)e^{\beta \!\int_\mathcal{Y} \!dy\,p(y| x)\log r(y| t)}\right),\\
    \end{aligned}
    \label{dual}
\end{equation}
which consists of two variables and can be viewed as the dual form\footnote{We adopt this designation, due to its similarity with the rate function duality \cite[Sec. 2]{dembo2002source}, \cite[Lemma 2]{lei2023neural} in the rate-distortion theory.} of the IB problem.

%

%
Then, we complete the proof from the perspective of mapping. 
We state that given continuous random variables $T$ and $Z$, with their probability measures $Q_T$ and $P_Z$, there exists a mapping $\varphi:\mathcal{Z}\rightarrow\mathcal{T}$, such that $T=\varphi(Z)$ and $Q_T=\varphi_{\#}P_Z$, where $\varphi_{\#}$ is the {push-forward} operator defined by
\begin{equation}
    \forall h\in\mathcal{C}(\mathcal{T}) , \int_\mathcal{T}dQ_T\,h(t)=\int_{\mathcal{Z}}dP_Z\,h(\varphi(z)),
\end{equation}
and $\mathcal{C}(\mathcal{T})$ refers to the set of continuous functions on $\mathcal{T}$.

The existence of the mapping is guaranteed by the Borel isomorphism theorem \cite{royden1968real,gray2009probability}, i.e., a complete separable metric space with a finite Borel measure is isomorphic to the unit interval with Lebesgue measure. 
Since the measure spaces $(\mathcal{T},Q_T)$ and $(\mathcal{Z},P_Z)$ are both isomorphic to the unit interval with Lebesgue measure, these two spaces are isomorphic to each other. Therefore, the above statement is valid.

We now establish a one-to-one correspondence between probability distribution and mapping. 
For the dual form \eqref{dual}, we transform the problem of finding an optimal distribution $q(t)$ into searching an optimal mapping $\varphi$. 
As a result, we turn to minimize the functional 
\begin{equation*}
    \begin{aligned}
        &\min_{\varphi}\!\min_{r(y| t)} G(\varphi,r(y| t))=\\
        - &\!\!\int_\mathcal{X}\!\! dx\,p(x)\log\left(\int_{\mathcal{Z}}\!\!dz\,p(z)e^{\beta \int_\mathcal{Y} \!dy\,p(y| x)\log r(y| \varphi(z))}\right).\\
    \end{aligned}
\end{equation*}

It is noticed that the two variables $\varphi$ and $r(y|t)$ exist in the functional $G$ in the composite form of $r(y|\varphi(z))$.
Considering the fact that $\varphi$ is a mapping from $\mathcal{Z}$ to $\mathcal{T}$, the expression $r(y|\varphi(z))$ is equal to the conditional probability density from $Z$ to $Y$, i.e., $r(y|z)$. 
By substituting $r(y|\varphi(z))=r(y|z)$, the functional $G$ depends solely on the variable $r(y|z)$.
To this end, the formulation \eqref{vmaib} is obtained. 
%
Since the above derivation only involves variable substitutions, the optimization problems \eqref{tri} and \eqref{vmaib} have equal optimal values.
\end{proof}

\begin{remark} 
The essence of mapping approach is the Borel isomorphism theorem \cite{royden1968real}, that a complete separable metric space with a finite Borel measure is isomorphic to the unit interval with Lebesgue measure. 
The isomorphism shows that the corresponding point-to-point mapping exists and reveals the correspondence between the problem of searching the optimized probability measure and the problem of finding optimal mapping from two spaces with measures. 
%
\end{remark}

\begin{remark}
The proof process reveals that the relation between the formulations \eqref{tri} and \eqref{vmaib} is not only about the equality of optimal values.
The relationship of their optimal solution is given by the mapping. 
Supposing $T$ is the variable generating $q(t)$ and $r(y|t)$ that minimizes \eqref{tri} and $Z$ optimizes \eqref{vmaib}, there exists a mapping $\varphi:\mathcal{Z}\rightarrow\mathcal{T}$ such that $T=\varphi(Z)$.
\end{remark}

\section{Neural Estimation For High-Dimensional Information Bottleneck Problems}

This section parameterizes the proposed MA-IB model using neural networks.
The MA-IB model possesses simpler formulations than existing methods, with less variables suitable for neural network parameterization, and without any relaxation capable for accurate estimation of the IB problem. 
We perform neural estimation on the MA-IB model similar to the methods proposed in \cite{lei2023neural}. 
Due to different optimization objectives, we parameterize the variable as a classifier instead of a generator as in \cite{lei2023neural}.
The proposed method is referred to as the neural MA-IB algorithm, which is notably applicable in high-dimensional scenarios.
Additionally, the asymptotic property of the proposed neural MA-IB algorithm is analyzed. 

\subsection{Neural Estimation }

The functional $G$ in \eqref{vmaib} can be rewritten in the expectation form as 
\begin{equation*}
G(r(y| z))=
            -E_{P_X} \!\!\left[\log\left(E_{P_Z}\!\!\left[e^{\beta E_{P_{Y|X}}[\log r(Y| Z)]}\right]\right)\right].
\end{equation*}

The expectation of high-dimensional distributions is intractable, and the Monte Carlo approximation is commonly used. 
%
Let $x_1,x_2,\ldots,x_n$ and $z_1,z_2,\ldots,z_m$ be samples i.i.d. from $P_X$ and $P_Z$ respectively, 
where $P_X$ is from a given dataset and $P_Z$ is a base distribution such as the high-dimensional standard Gaussian distribution.
{For each sample point $x_i$, we draw $y_{1,i},y_{2,i},\ldots,y_{l,i}$ i.i.d. from $P_{Y|X=x_i}$.
We can approximate $P_X$, $P_Z$ and $P_{Y|X=x_i}$ with the empirical data distribution. 
As a consequence, the expectation form can be transformed to
\begin{equation*}
\tilde{G}(r(y| z))=
    -\frac{1}{n}\sum_{i=1}^n\log\left(\frac{1}{m}\sum_{j=1}^me^{\beta\left(\frac{1}{l}\sum_{k=1}^l\log r(y_{k,i}| z_j)\right)}\right).
\end{equation*}
%
}

Next, we parameterize the variable $r(y| z)$ as  $r_\theta(y| z)$ with a network consisting of a multi-layer perception (MLP) followed by a softmax layer. 
%
%
%
{Substituting $r(y|z)$ by $r_\theta(y|z)$, one obtains the loss function
\begin{equation}
    F(\theta)=-\frac{1}{n}\sum_{i=1}^n\log\left(\frac{1}{m}\sum_{j=1}^me^{-\beta \kappa_{ij}(\theta)}\right),
    \label{loss}
\end{equation}
where $\kappa_{ij}(\theta)=-\frac{1}{l}\sum_{k=1}^l\log r_{\theta}(y_{k,i}| z_j)$.
}

In practice, we optimize the loss function by stochastic gradient descent. 
Algorithm 1 summarizes the proposed neural MA-IB estimation method. 

%
\begin{algorithm}[ht]
\renewcommand{\algorithmicrequire}{\textbf{Input:}}
\renewcommand{\algorithmicensure}{\textbf{Output:}}
\renewcommand{\algorithmicreturn}{\textbf{Return:}}
\caption{Neural MA-IB}\label{alg:NEIB}
\begin{algorithmic}
    \REQUIRE sample sizes $m$, $n$ and $l$, learning rate $\eta$, $max\_iter$
    \STATE initialize the parameters $\theta$ of the network $r_\theta(y| z)$
    \FOR{$r=1:max\_iter$}
        \STATE sample $\{x_i\}^n_{i=1}$  i.i.d. from $P_X$ 
        \STATE sample $\{z_j\}^m_{j=1}$  i.i.d. from $P_Z$
        \STATE sample $\{y_{k,i}\}^l_{i=1}$  i.i.d. from $P_{Y|X=x_i}$
        \STATE calculate $\kappa_{ij}(\theta)=-\frac{1}{l}\sum_{k=1}^l\log r_{\theta}(y_{k,i}| z_j)$ 
        \STATE calculate $F(\theta)=-\frac{1}{n}\sum_{i=1}^n\log\left(\frac{1}{m}\sum_{j=1}^me^{-\beta \kappa_{ij}(\theta)}\right)$
        \STATE update $\theta$ by gradient descent $\theta\leftarrow \theta-\eta \nabla_\theta F(\theta)$
    \ENDFOR
    \RETURN $r_\theta(y| z)$
\end{algorithmic}
\end{algorithm}

\subsection{Asymptotic Analysis}

This subsection discusses the relation  between the optimal value of the loss function \eqref{loss} and that of the optimization problem \eqref{vmaib}, which is precisely defined by the following theorem\footnote{In work \cite{belghazi2018mutual} and \cite{lei2023neural}, this property is also called strong consistency.}.

\begin{theorem}
Suppose $G$ is the optimal value of \eqref{vmaib}, and then there exists a compact parameter domain $\Theta$, such that the optimal value of $F(\theta)$, i.e.
\begin{equation*}
    \widehat{G}_{mnl}=\inf_{\theta\in \Theta}F(\theta),
\end{equation*}
 satisfies that $\widehat{G}_{mnl}$ converges to $G$ with probability one.
Specifically, for all $\epsilon>0$, there exist positive  integers $M,N,L$, such that $\forall n>N,m>M,l>L$,
\begin{equation}
    |G-\widehat{G}_{mnl}|<\varepsilon,\quad a.e.
\label{ae}
\end{equation}
where “a.e.” stands for almost everywhere.
\end{theorem}

The proof of this theorem includes two parts. It involves an approximation from the neural network space to the entire function space, and an approximation from the empirical distribution to the true distribution. 
These two parts are guaranteed by the universal approximation theorem \cite{hornik1989multilayer,kidger2020universal} and the law of large numbers \cite{gray2009probability}. 
%
\begin{proof}
    Fix $\varepsilon>0$. Let $r^*(y| z)$ be the optimal solution that minimizes \eqref{vmaib}, and then $G$ satisfies
    \begin{equation*}
            G=
            -E_{P_X}\!\!\left[\log\left(E_{P_Z}\!\!\left[e^{\beta E_{P_{Y|X}}[\log r^*(Y| Z)]}\right]\right)\right].
    \end{equation*}
    Next, define 
    \begin{equation*}
        G_{\theta}=-E_{P_X} \!\!\left[\log\left(E_{P_Z}\!\!\left[e^{\beta E_{P_{Y|X}}[\log r_\theta(Y| Z)]}\right]\right)\right].
    \end{equation*}
    
    Fix $\xi>0$. By the universal approximation theorem \cite{hornik1989multilayer,kidger2020universal}, there exist $\Theta$ and $\hat{\theta}\in\Theta$, such that $\forall y\in\mathcal{Y},z\in\mathcal{Z}$,
    $$
    |\log r^*(y| z)-\log r_{\hat{\theta}}(y| z)|<\xi.
    $$
    Then $\text{for } \forall\,x\in\mathcal{X}\text{ and }z\in\mathcal{Z}$, it holds that
    \begin{equation*}
    \begin{aligned}
        &\left|\beta E_{P_{Y|X=x}}[\log r^*(Y| z)]-\beta E_{P_{Y|X=x}}[\log r_\theta(Y| z)]\right|\\
        &\leq\beta E_{P_{Y|X=x}}\left|\log r^*(Y| z)-\log r_{\hat{\theta}}(Y| z)\right|<\beta\xi.
    \end{aligned}
    \end{equation*}
    
    
    Considering that $e^x$ is Lipschitz continuous on the interval $(-\infty,0]$ with constant 1 and $\log x$ is a continuous function, there exists a positive number $\xi$, such that the above relation yields
    \begin{equation*}
        \left|\log \frac{E_{P_Z}\!\!\left(e^{\beta E_{P_{Y|X=x}}[\log r^*(Y| Z)]}\right)}{E_{P_Z}\!\!\left(e^{\beta E_{P_{Y|X=x}}[\log r_\theta(Y| Z)]}\right)}\right|<\frac{\varepsilon}{2}.
    \end{equation*}

    By choosing a suitable $\hat{\theta}$, one has $|G-G_{\hat{\theta}}|<\frac{\varepsilon}{2}$.
    Next, define
    \begin{equation*}
        G_{\Theta}=\inf_{\theta\in\Theta}G_\theta.
    \end{equation*}
    From this definition, one obtains $G\leq G_{\Theta} \leq G_{\hat{\theta}}$, and then
    \begin{equation}
        |G-G_{\Theta}|<\frac{\varepsilon}{2}.
        \label{ie1}
    \end{equation}
    %
    
    Further, one has 
    $$
    |\widehat{G}_{mnl}-G_{\Theta}|
    =|\inf_{\theta\in\Theta}F(\theta)-\inf_{\theta\in\Theta}G_\theta|
    \leq\sup_{\theta\in\Theta}|F(\theta)-G_\theta|.$$
    Since $\Theta$ is compact and the neural networks are continuous, 
    $$
    f_\theta(X,Z):=\log r_\theta(y| Z)
    $$ 
    satisfies the law of large numbers \cite{gray2009probability}. 
    Therefore, for any given $\epsilon>0$, there exist $N$,$M$ and $L$, such that 
    $\forall n>N,m>M,l>L$ and with probability one, $|F(\theta)-G_\theta|\leq\frac{\varepsilon}{2}$.
    Therefore
    \begin{equation}
        |\widehat{G}_{mnl}-G_{\Theta}|<\frac{\varepsilon}{2}.
        \label{ie2}
    \end{equation}
    Combining \eqref{ie1} and \eqref{ie2}, one finally has
    \begin{equation*}
        |G-\widehat{G}_{mnl}|<\varepsilon,\,\forall\,n>N,m>M,l>L.
    \end{equation*}
    The proof of the theorem has completed here. 
\end{proof}

\section{Experimental Results and Discussions}
This section evaluates the effectiveness of the proposed neural MA-IB algorithm, by examining over both low-dimensional cases and high-dimensional MNIST dataset. 

\subsection{Low-Dimensional Examples}

This subsection computes IB curves of two low-dimensional models: a toy model and the joint Gaussian model.

For the toy model, we set the joint distribution of $X$ and $Y$ as
$$
\left[p(x_i,y_j)\right]=\frac{1}{90}
\left[\begin{array}{cccc}
     1 & 2 & \cdots & 9 \\
     9 & 8 & \cdots & 1 \\
\end{array}\right]^T.
$$

For the joint Gaussian model, we set $Y = X+E$, where $X$ and $E$ independently follow the  standard Gaussian distribution $\mathcal{N}(0,1)$. 
Since the joint Gaussian distribution is continuous, we should discretize $Y$ before applying the neural MA-IB algorithm. 
Truncate $Y$ into an interval $[-M,M]$ and divide the interval into $N$ equal parts. Then record the midpoints of these pieces as $y_j=-M+\delta (j-1/2),\,j=1,\dots,N$, where $\delta = 2M/N$. 
For each sampling point $x_i$, the conditional probability is
$$
    p(y_j| x_i)=\frac{1}{\sqrt{2\pi}}e^{-(y_j-x_i)^2/2},\,j=1,\ldots,N,
$$
and the discretized distribution is $\hat{p}(y_j| x_i)=p(y_j| x_i)\delta$.
In this paper, $M=10$ and $N=100$ are employed.

\vspace{-.1in}
\begin{figure}[H]
    \centering
    \includegraphics[width=0.485\linewidth]{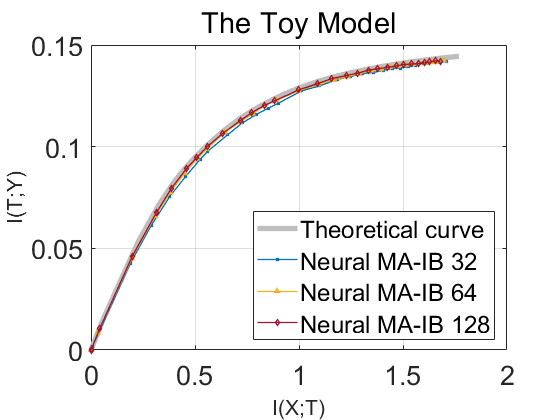}
    \label{toy}
    \includegraphics[width=0.485\linewidth]{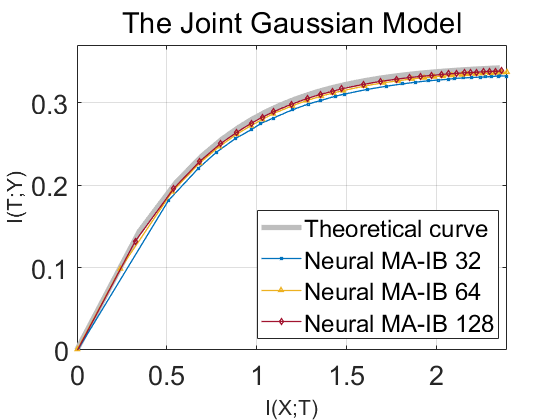}
    \label{gas}
    \caption{Comparison of the theoretical result and the neural MA-IB algorithm with different sample sizes on the toy model (left) and the joint Gaussian model (right). The grey line is the theoretical curve and the blue, yellow, and red lines are obtained by neural MA-IB algorithm with sample sizes of 32, 64, and 128, respectively.}
    \label{fig:1} 
\end{figure}
\vspace{-.1in}

Fig. \ref{fig:1} plots the theoretical IB curve drawn by the BA algorithm and the curves produced by the proposed neural MA-IB algorithm, for the toy model (left) and joint Gaussian model (right). 
It is observed that the curves move closer to the theoretical curve with increasing sample sizes, which demonstrates the effectiveness of neural MA-IB algorithm and its asymptotic property.
%

\subsection{Application in High-Dimensional Dataset}
This subsection applies the neural MA-IB algorithm into the real-world MNIST dataset \cite{lecun2010mnist}. 
The MNIST dataset is a digit classification dataset, where $X$ represents $28\times28$ grayscale image data, and $Y$ corresponds to the labels. 
Due to the dimension of $X$ being as high as 768, the corresponding IB problem is high-dimensional and cannot be computed by numerical algorithms.

As MNIST is a classification dataset, it holds that $p(y| x_i)=\delta_{y_i}$, and thus, $Y$ is a deterministic function of $X$. 
In this case, the authors of \cite{tishby2003tradeoff} proved that the corresponding IB curve is piece-wise linear, i.e., 
$$
I(R)=\left\{ \begin{array}{l}
     R\,,\,R\leq H(Y)  \\
    H(Y)\,,\,R>H(Y)
\end{array}
\right..
$$

Fig. \ref{fig:2} compares the theoretical IB curve with the curves generated by the neural MA-IB algorithm and the VIB method. 
It is obvious that the points obtained by the neural MA-IB algorithm almost coincides with the theoretical curve. 
In contrast, the VIB method only provides a coarse estimation of the desired IB curve. 
Note that it is reasonable for the points given by the neural MA-IB algorithm to cluster near the inflection point $(H(Y ),H(Y))$ \cite{kolchinsky2018caveats}, since the neural MA-IB algorithm fixes the Lagrange multiplier $\beta$.

\begin{figure}[H]
    \centering
    \includegraphics[width=0.8\linewidth]{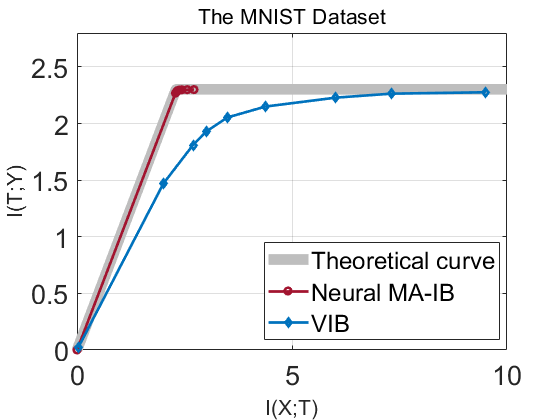}
    \caption{Comparison of the IB curves for the MNIST dataset, with VIB method (blue), neural MA-IB method (red) and the theoretical result (grey).}
    \label{fig:2}
\end{figure}

\section{Conclusion}
In this work, we propose a novel approach to solving IB problem. 
We reformulate the IB problem from the perspective of mapping, through which we transform the original problem into an equivalent simplified model, i.e., the MA-IB model. 
We then propose neural MA-IB algorithm that utilizes neural networks to solve the derived MA-IB model. 
Further, the asymptotic analysis of the proposed algorithm is conducted.
Experimental results show the effectiveness of the proposed neural MA-IB algorithm for both low-dimensional and high-dimensional IB problems.


\bibliographystyle{bibliography/IEEEtran}
\bibliography{bibliography/MAIB_REF}


\end{document}